\documentclass{amsart}
\usepackage{amsfonts}
\usepackage{amsmath}
\usepackage{amssymb}
\usepackage{amsthm}
\usepackage{color}
\usepackage{comment}

\usepackage[foot]{amsaddr}

\usepackage{pdfpages}

\newtheorem{Thm}{Theorem}
\newtheorem{Lem}[Thm]{Lemma}
\newtheorem{Prop}[Thm]{Proposition}

\theoremstyle{definition}
\newtheorem{Def}[Thm]{Definition}

\theoremstyle{remark}
\newtheorem{Rem}[Thm]{Remark}

\numberwithin{equation}{section}

\def\RR{{\mathbb{R}}}

\def\bigO{\mathcal{O}}

\title[Don't cry to be the first! Symmetric fair division algorithms exist]{Don't cry to be the first!\\
Symmetric fair division algorithms exist.}

\date{\today}

\author[G.~Ch\`eze]{Guillaume Ch\`eze}
\address{Guillaume Ch\`eze: Institut de Math\'ematiques de Toulouse\\
Universit\'e Paul Sabatier \\
118 route de Narbonne\\
31 062 TOULOUSE cedex 9, France}
\email{guillaume.cheze@math.univ-toulouse.fr}

\begin{document}


\begin{abstract}
In this article we study a cake cutting problem. More precisely, we study \emph{symmetric} fair division algorithms, that is to say we study algorithms where  the order of the players does not influence the value obtained by each player.
In the first part of the article, we give a symmetric and envy-free fair division algorithm. More precisely, we show how to get a symmetric and envy-free fair division algorithm from an envy-free division algorithm. \\
In the second part, we  give a proportional and symmetric fair division algorithm with a complexity in $\bigO(n^3)$ in the Robertson-Webb model of complexity. This algorithm is based on Kuhn's algorithm. Furthermore, our study has led us to study aristotelian fair division. This notion is an interpretation of Aristotle's principle: give equal shares to equal people.\\
We conclude this article with a discussion and some questions about the Robertson-Webb model of computation.
\end{abstract}


\maketitle

\section*{Introduction}
In this article we study the problem of fair resource allocation.  It consists to share an heterogeneous good between different players or agents. This good can be for example: a cake, land, time or computer memory. This problem is old. For example, the Rhind mathematical papyrus contains problems about the division of loaves of bread and about partition of plots of land. In the Bible we find the famous ``Cut and Choose" algorithm and in the greek mythology we find the trick at Mecone.\\
  The problem of fair division has been formulated in a scientific way by Steinhaus in 1948, see \cite{Steinhaus}. Nowadays, there exists several papers, see e.g. \cite{DubinsSpanier, EvenPaz, EdmondsPruhs, BramsTaylorarticle, RoberstonWebbarticle, Pikhurko, Thomson2006, Procacciasurvey, BJK, AzizMackenzie}, and books about this topic, see e.g. \cite{RobertsonWebb,BramsTaylor, Procacciachapter,Barbanel}. These results appear in the mathematics, economics, political science, artificial intelligence and computer science literature. Recently, the cake cutting problem has been studied intensively by computer scientists for solving resource allocation problems in multi agents systems, see e.g.~\cite{Chevaleyre06,Chen,Dynamic,Branzei}. \\
 
 Throughout this article, the cake will be an heterogeneous good represented by the interval $[0,1]$.
We consider $n$ players and we associate to each player a non-atomic probability measure $\mu_i$ on the interval $X=[0,1]$. These measures represent the utility functions of the player. The set $X$ represents the cake and we have $\mu_i(X)=1$ for all $i$. The problem in this situation is to get a fair division of $X=X_1\sqcup \ldots \sqcup X_n$, where the $i$-th player get $X_i$.\\

A practical problem is the computation of fair divisions. In order to describe algorithms we thus need a model of computation. There exist two main classes of cake cutting algorithms: discrete and continuous protocols (also called moving knife methods). Here, we study discrete algorithms. These kinds of algorithms can be  described thanks to the  classical model introduced by Robertson and Webb and formalized by Woeginger and Sgall in \cite{Woeg}. In this model we suppose that a mediator interacts with the agents. The mediator asks two type of queries: either cutting a piece with a given value, or evaluating a given piece. More precisely, the two type of queries allowed are:
\begin{enumerate}
\item $eval_i(x,y)$: Ask agent $i$ to evaluate the interval $[x,y]$. This means return $\mu_i([x,y])$.
\item $cut_i(x,a)$: Ask agent $i$ to cut a piece of cake $[x,y]$ such that $\mu_i([x,y])=a$. This means: for given $x$ and $a$, return $y$ such that $\mu_i([x,y])=a$.
\end{enumerate} 
In the Robertson-Webb model the mediator can adapt the queries from the previous answers given by the players. In this model, the complexity counts the finite number of queries necessary to get a fair division. For a rigourous description of this model we can consult: \cite{Woeg,Branzei2017}.\\

When we design a cake cutting algorithm, we have to precise what is the meaning of a fair division. Indeed, there exists different notions of fair division.\\ 
We say that a division is \emph{proportional} when for all $i$, we have $\mu_i(X_i) \geq 1/n$.\\
We say that a division is \emph{envy-free} when for all $i \neq j$, we have $\mu_i(X_i) \geq  \mu_i(X_j)$.\\
We say that a division is \emph{equitable} when for all $i\neq j$, we have $\mu_i(X_i)=\mu_j(X_j)$.\\

The first studied notion of fair division has been proportional fair division, \cite{Steinhaus}. Proportional fair division is a simple and well understood notion. In \cite{Steinhaus} Steinhaus explains the Banach-Knaster algorithm, also called last diminisher algorithm, which gives a proportional fair division. There also exists an optimal algorithm to compute a proportional fair division in the Robertson-Webb model, see \cite{EvenPaz,EdmondsPruhs}. The complexity of this algorithm is in $\bigO\big(n\log(n)\big)$. Furthermore, the portion $X_i$ given to the $i$-th player in this algorithm is an interval.\\

It is more difficult to get an envy-free fair division. Indeed, whereas envy-free fair divisions where each $X_i$ is an interval exist, there does not exist an algorithm in the Robertson-Webb model computing such divisions. These   results have been proved by Stromquist in \cite{Stromquistexist,Stromquist}. The first envy-free algorithm has been given by Brams and Taylor in  \cite{BramsTaylorarticle}. This algorithm has been given  approximatively 50 years after the first algorithm computing a proportional fair division. The Brams-Taylor algorithm has an unbounded complexity in the Robertson-Webb model. This means that we cannot bound the complexity of this algorithm in terms of the number of players only.  It is only recently that a finite and unbounded algorithm has been given to solve this problem \cite{AzizMackenzie}. The complexity of this algorithm is in $\bigO\Big(n^{n{^{n^{n^{n^{n}}}}}}\Big)$. A lower bound for envy-free division algorithm has been given by Proccacia in \cite{Procaccia-lowerbound}. This lower bound is in $\bigO(n^2)$.\\

Equitable fair divisions have been less studied than proportional and envy-free divisions. However, there exist some results showing the difficulty to get such fair divisions. Indeed, there exist equitable fair divisions where each $X_i$ is an interval, see \cite{Cechexistence,Segal-Halevi,Chezeequitable}. However, there do not exist algorithms computing an equitable fair division, see \cite{Cech,ProcWang,ChezeBSSRW}. \\

In practice, a cake cutting algorithm $\mathcal{F}$ has in inputs a list of measures $\underline{\mu}=[\mu_1,\ldots,\mu_n]$, and returns a partition 
$X=\mathcal{F}(X,\underline{\mu},1) \sqcup \ldots \sqcup \mathcal{F}(X,\underline{\mu},n)$, where each $\mathcal{F}(X,\underline{\mu},i)$ is a finite union of disjoint intervals.  The set $\mathcal{F}(X,\underline{\mu},i)$ is the part given to the $i$-th player appearing in the the list $\underline{\mu}$ when we apply the algorithm $\mathcal{F}$ to this list of measures.\\

The definition of proportional or envy-free fair division is independent of the order of the players in the list $\underline{\mu}$. However, this order is  important in cake-cutting algorithms. For example, the role of the two players in the ``Cut and Choose" algorithm are not symmetric. This leads the definition of \emph{symmetric fair division algorithm}.\\
\begin{Def}
We denote by $\underline{\mu}^{\sigma}$ the list $\underline{\mu}^{\sigma}=[\mu_{\sigma(1)},\ldots,\mu_{\sigma(n)}]$, where $\sigma$ belongs to the permutation group  $\mathfrak{S}_n$. 
A cake cutting algorithm $\mathcal{F}$ is \emph{symmetric}  when 
$$\forall i \in \{ 1, \ldots, n\}, \forall \sigma \in \mathfrak{S}_n, \,  \mu_i\big(\mathcal{F}(X,\underline{\mu},i)\big)=\mu_i\big(\mathcal{F}(X,\underline{\mu}^{\sigma},\sigma^{-1}(i))\big).$$
\end{Def}
For example, if $n=3$ and $\sigma=(1\,2\,3)$ then a symmetric fair division algorithm satisfies:
$$\mu_1\big(\mathcal{F}(X,[\mu_1,\mu_2,\mu_3],1)\big)=\mu_1\big(\mathcal{F}(X,[\mu_2,\mu_3,\mu_1],3)\big).$$
A cake cutting algorithm is symmetric means  whatever the order of the measure given in inputs, all  players will receive the same value of the cake. Indeed, $\mathcal{F}(X,[\mu_2,\mu_3,\mu_1],3)$ is the portion given to the third player in the list $[\mu_2,\mu_3,\mu_1]$. Thus, this corresponds to the portion given to the player with measure $\mu_1$ when the algorithm $\mathcal{F}$ as in input the list  $[\mu_2,\mu_3,\mu_1]$. Thus, if the player with associated measure $\mu_1$ is in the first or in the last position in the inputs he or she will get a portion with the same measure relatively to his or her preference $\mu_1$. Therefore, there is no advantage to be the first in the list $\underline{\mu}$. The measure of the received  portion is independent of the position of a player in the list.\\
 This notion has been introduced by Manabe and Okamoto in \cite{ManabeOkamoto}. They call this kind of fair division  \emph{meta envy-free}. In this article we call this property \emph{symmetric} in order to emphasize the role of the permutations of the players. In their paper Manabe and Okamoto have shown that classical algorithms such as Selfridge-Conway, and Brams-Taylor's algorithms are not symmetric. Then they
 have given a symmetric and envy-free algorithm for 4 players and ask if it is possible to get such a division protocol for $n\geq 4$ players. Here, we answer to this question and we prove the following result:
\begin{Thm}
There exists  deterministic symmetric  and envy-free cake cutting algorithms.
\end{Thm}
In order to prove this result we show how to construct such an algorithm from an envy-free algorithm. The idea is to use an already existing  envy-free algorithm $f$, see e.g. \cite{BramsTaylorarticle, RoberstonWebbarticle,Pikhurko,AzizMackenzie} and to construct from it a symmetric and envy-free algorithm $\mathcal{F}$. In order to get a symmetric algorithm we compute all $f(\underline{\mu}^{\sigma})$ and then we take the ``best" one. Here ``best" will mean : satisfy some topological conditions, e.g. we select a partition with the minimal number of cuts.\\
 
 Our approach computes  $n!$ envy-free divisions, thus  this gives an algorithm with a huge complexity in the Robertson-Webb model. Furthermore, our algorithm gives  a proportional division since it gives an envy-free division. A natural question is then: Can we get a symmetric and proportional division algorithm with a polynomial complexity?\\
We prove in Section \ref{sec:propvrac} the following result:
 
 \begin{Thm}
 There exists a deterministic symmetric and proportional  algorithm which uses at most $\bigO(n^3)$ queries in the Robertson-Webb model.
 \end{Thm}
  
  The deterministic assumption is important. We do not want to get a situation where a player could think that he is unlucky.\\
  We can already remark that the Evan-Paz algorithm, see \cite{EvenPaz}, and the last diminisher procedure are not deterministic and not symmetric. Indeed, if during these algorithms several players cut the cake at the same point, then this tie is usually breaked with a random process. Another way to break the tie is to use the order on the players. For example, if all the players in the first step of the Evan-Paz algorithm cut the cake at the same point, then  we can give to the players $1, \ldots, \lfloor n/2 \rfloor$ the left part of the cake and to the other players the right part of the cake. This tie breaking method depends on the order the players and thus it does not give a symmetric procedure.\\

At last, in this article we study also another fair division notion. This notion comes from the study of symmetric fair divisions in a particular case: Suppose that $\mathcal{F}$ is a symmetric fair division algorithm. Then we have
$$\mu_1\big( \mathcal{F}(X,[\mu_1,\mu_2,\mu_3],1) \big) = \mu_1\big( \mathcal{F}(X,[\mu_2,\mu_1,\mu_3],2) \big).$$
Now, suppose  that $\mu_1=\mu_2$, this gives
$$\mu_1\big( \mathcal{F}(X,[\mu_1,\mu_2,\mu_3],1) \big) = \mu_2\big( \mathcal{F}(X,[\mu_1,\mu_2,\mu_3],2) \big).$$
This means that if two players have the same measure then they consider as equal the portions they get. We call a fair division satisfying this property an ``\emph{aristotelian fair division}".

\begin{Def}
We say that we have an aristotelian division when $\mu_i=\mu_j$ implies  $\mu_i(X_i)=\mu_j(X_j)$.
\end{Def}

We have given the name ``aristotelian fair division" to this kind of fair divisions because in the  Nicomachean Ethics by Aristotle (Book~V) we find:\\

\emph{ ``\ldots it is when equals possess or are allotted unequal shares, or persons not equal equal shares, that quarrels and complaints arise."}\\

Therefore, aristotelian fair division is not a new notion. This notion has been already studied, see e.g.~\cite{Maniquet,MoulinBook,MoulinSym}. In the literature this notion also appears as ``Equal Treatment of Equals".\\

We remark that symmetric fair division algorithms give aristotelian fair divisions. However, the converse is not true.\\

As a first step towards the construction of a symmetric and proportional fair division algorithm, we describe in Section~\ref{sec:propvrac} an aristotelian and proportional fair division algorithm. This algorithm needs $\bigO(n^3)$ queries but less arithmetic operations than the symmetric and proportional algorithm.\\
 We remark easily  that 
  an envy-free division is always proportional and aristotelian, but a fair division which is aristotelian and proportional  is less demanding than an envy-free division. However, to the author's knowledge all existing aristotelian proportional fair division algorithms were envy-free algorithms.\\
Thus our algorithm shows that if we just want an aristotelian proportional fair division it is not necessary to use an envy-free algorithm which uses an exponential number of queries.

\subsection*{Structure of the paper} In Section~\ref{sec:symenvyfree}, we give a symmetric and envy-free fair division algorithm. Then, we give some remarks about the complexity of this algorithm. In this first section, we also discuss the problem of symmetric and envy-free fair division in the approximate setting. In Section~\ref{sec:propvrac}, we explain why the Evan-Paz and the last diminisher algorithm do not give aristotelian fair division. Then we give an aristotelian proportional fair division algorithm and next a symmetric and proportional fair division algorithm. In Section~\ref{sec:conclusion}, we conclude this article with several questions about symmetric and aristotelian fair divisions and the Robertson-Webb model of computation.

\section{An envy-free and symmetric cake cutting algorithm}\label{sec:symenvyfree}
\subsection{Two orders on partitions and one algorithm}
In this section we introduce two different orders on the partitions. These orders will be used to choose a ``good" partition among the $n!$ possible fair divisions given by all  $f(\underline{\mu}^{\sigma})$, where $f$ is a fair division procedure.\\

In this section, when we study a partition $X=X_1 \sqcup \ldots \sqcup X_n$, $X_i$ will be the part given to the $i$-th player.\\

For each  partition $X=X_{1} \sqcup \ldots \sqcup X_{n}$ we set 
$$X_{i}=\bigsqcup\limits_{j \in I_i} [x_{i,j},x_{i,j+1}],   \textrm{ where } I_i \textrm{ is a finite set.}$$
Thus $$X=\bigsqcup\limits_{i=1}^{n} \bigsqcup\limits_{j \in I_i} [x_{i,j},x_{i,j+1}]$$
and  
$$X=\bigsqcup\limits_{l=0}^{M}[z_l,z_{l+1}]$$
 where $z_0=0$, $z_{M+1}=1$, $z_l=x_{i,j}$ and $z_l < z_{l+1}$. From this partition we construct a  vector $(z_1,\ldots,z_M) \in \RR^M$. We say that  $M+1$ is the size of the partition.
\begin{Def}
The graded order on $\sqcup_{k=1}^{\infty} \RR^k$ is the following:\\
Let $(x_1,\ldots,x_M) \in \RR^M$ and $(y_1,\ldots, y_N) \in \RR^N$ we have:
\begin{eqnarray*}
(y_1,\ldots, y_N) \succ_{gr} (x_1,\ldots,x_M) &\iff & N>M\\
&&\textrm{ or } N=M \textrm{ and } y_1>x_1,\\
&& \textrm{ or }N=M,  \exists j> 1 \textrm{ such that } y_i=x_i \textrm{ for  } i< j\\
&& \hphantom{ or } \textrm{ and } y_j >  x_j.
\end{eqnarray*}
\end{Def}

The graded order gives thus an order on the partitions.\\

Now, we give an algorithm which computes a word over the alphabet $a_1,\ldots,a_n$ from a partition. The $l$-th letter of the word $\omega$ is denoted by $\omega(l)$.\\

\texttt{Word from partition}\\
\textsf{Input:} A partition $X=X_1\sqcup \ldots \sqcup X_n$, where $X_i=\sqcup_{j \in I_i} [x_{i,j},x_{i,j+1}]$, and \\
$X=\sqcup_{l=1}^{M}[z_l;z_{l+1}]$ is the associated decomposition.\\
\textsf{Output:} A word $\omega$ constructed over the alphabet $a_1,\ldots,a_n$.

\begin{enumerate}
\item If  $[z_0,z_1] \subset X_j$ then $a_1$ is associated to $X_j$ and $\alpha:=2$.
\item $\omega(1):=a_1$.
\item For $l$ from 1 to $M$ do
\begin{enumerate}
\item[] If $[z_l,z_{l+1}] \subset X_i$ and $X_i$ is associated to $a_k$ where $k<\alpha$ \\Then $\omega(l+1):=a_k$, \\
Else 
associate $a_{\alpha}$  to $X_i$, $\omega(l+1):=a_{\alpha}$, and $\alpha:=\alpha+1$.\\
\end{enumerate}
\end{enumerate}

Now, we introduce a second order on the partitions.

\begin{Def}
Consider two partitions $X=X_1\sqcup \ldots \sqcup X_n$ and $X=X'_1\sqcup \ldots \sqcup X'_n$. With the previous algorithm we associate a word $\omega$ to the first partition and we associate a word $\omega'$ to the second partition. 

If $\omega \succ_{lex} \omega'$, that is to say,  if $\omega$ is bigger than $\omega'$ with the lexicographic order with $a_n \succ_{lex} a_{n-1} \succ_{lex} \ldots \succ_{lex} a_1$,  then we say that the partition $X=X_1\sqcup \ldots \sqcup X_n$ is bigger than the partition $X=X'_1\sqcup \ldots \sqcup X'_n$ relatively to the lexicographic order.\\
If two partitions gives the same word then we say that the partitions are equal relatively to the lexicographic order.
\end{Def}

\begin{Lem}\label{lem:permut}
Consider two partitions $X=X_1\sqcup \ldots \sqcup X_n$ and $X=X'_1\sqcup \ldots \sqcup X'_n$. If these partitions give the same vector $(z_1,\ldots,z_M)$ and if these partitions are  equal relatively to the lexicographic order, then  there exists a permutation $\sigma \in \mathfrak{S}_n$ such that: 
$$X_{\sigma(i)} =X'_{i}.$$
\end{Lem}

\begin{proof}
This follows from the construction of the lexicographic order on the partitions.
\end{proof}

The two previous orders allow us to get a symmetric and envy-free fair division.\\

\texttt{Symmetric and Envy-free}\\
\textsf{Inputs:} $\underline{\mu}=[\mu_1,\ldots,\mu_n]$, a deterministic envy-free cake cutting algorithm $f$.\\
\textsf{Outputs:} $X=\mathcal{F}(X,\underline{\mu},1) \sqcup \ldots \sqcup \mathcal{F}(X,\underline{\mu},n)$, where $\mathcal{F}(X,\underline{\mu},i)$ is a finite union of disjoint intervals and $\mathcal{F}(X,\underline{\mu},i)$ is given to the $i$-th player.

\begin{enumerate}
\item For all $\sigma \in \mathfrak{S}_n$, computes the partition  $f(\underline{\mu}^{\sigma})$ and\\ set $S:=\{f(\underline{\mu}^{\sigma}) \, | \, \sigma \in \mathfrak{S}_n \}$.
\item Let $S_1$ be the subset of $S$ of  all partitions with a minimal graded order.
\item  If $|S_1|=1$, then Return the unique partition in $S_1$, else go to the next step.
\item Let $S_2$ be the set of all the partitions in $S_1$ with a minimal lexicographic order.
\item \label{step:fin_envyfree} Return a partition $f(\underline{\mu}^{\sigma}) \in S_2$.
\end{enumerate}

\begin{Thm}
The  algorithm \emph{\texttt{Symmetric and Envy-free}} is deterministic symmetric and envy-free.
\end{Thm}

\begin{proof}
This algorithm is envy free because we return a result coming from an envy-free protocol.\\

We remark that if we apply the algorithm to the list $\underline{\mu}$ or $\underline{\mu}^{\rho}$ where $ \rho \in\mathfrak{S}_n$, then the set $S$ computed in the first step will always be the same. Therefore, we just have to study the situation where $S_2$ contains several partitions.\\ Consider two distinct partitions in $S_2$, $X=X_1\sqcup \ldots \sqcup X_n$ and $X=X'_1\sqcup \ldots \sqcup X'_n$. Thanks to Lemma \ref{lem:permut}, there exists a permutation $\sigma \in \mathfrak{S}_n$ such that $X_{\sigma(i)}=X'_{i}$.\\
As $f$ is an envy-free protocol, if the $i$-th player receives the portion $X_i$ then we have \mbox{$\mu_i(X_i) \geq \mu_i(X_{\sigma(i)})$}. Therefore, $\mu_i(X_i) \geq \mu_i(X_{\sigma(i)})=\mu_i(X'_i)$. In the same way, we show that $\mu_i(X'_i) \geq \mu_i(X_i)$. This gives $\mu_i(X'_i) =\mu_i(X_i)$. Then, for all partitions in $S_2$ each player will evaluate in the same way his or her portion. Thus the algorithm is symmetric. \\
In step \ref{step:fin_envyfree} we have to choose a partition among all  partitions in $S_2$. We can choose the first computed partition appearing in $S_2$. This last step depends on the order of the measures given in input. However, as  explained before this choice does not have en effect on how the $i$-th player evaluate his or her part.
\end{proof}

The idea of the algorithm is the following: if we have different possible partitions coming from all the $f(\underline{\mu}^{\sigma})$ then we prefer the ones with the fewest number of intervals and with the smallest leftmost part. It seems natural to prefer a partition with few intervals. The second condition can be interpreted as follows: If the different pieces of cake are given from left to right, thus in  increasing order of the  $x_{i,j}$, then our algorithm gives a first piece with small length to the first served player. If we imagine that a mediator is used to cut the cake then our convention means the following: if a player cooperates quickly with the mediator (the player accepts the leftmost part of the cake) then he gets quickly a piece of cake.

\subsection{Some remarks about the complexity of symmetric and envy-free algorithm}\label{sec:complexity}

Our algorithm relies on an envy-free division algorithm and needs to compute all  fair divisions for all permutation orders. Suppose that this envy-free division algorithm has a  complexity equals to $T(n)$ in the Robertson-Webb model, then our algorithm uses $n! \times T(n)$ queries. 
Indeed, our approach needs to compute all the fair divisions for all permutation orders. A natural question is the following: Is it necessary?\\

 Recently Aziz and Mackenzie have proposed in \cite{AzizMackenzie} the first envy-free algorithm with a complexity bounded in terms of the number of players.  If we use this algorithm then we get a symmetric and envy-free algorithm with a complexity bounded in terms of the number of players. \\
 
At last, we remark that if the envy-free algorithm $f$ uses a continuous protocol (a moving knife method) then our algorithm $\mathcal{F}$ gives a continuous protocol to compute a symmetric and envy-free division.
 \subsection{Approximate symmetric and envy-free fair division algorithm}\label{sub:sec-approx}
Envy-free fair division has also been studied in an approximate setting. A division is said to be $\varepsilon$-envy-free when we have for all $i$ and $j$: $\mu_i(X_i)\geq \mu_i(X_j) -\varepsilon$, where $\varepsilon >0$.
There exists an algorithm which gives such fair division, see \cite{Branzei2017}. The complexity of this algorithm is in $O(n/\varepsilon)$ in the Robertson-Webb model.\\
In the approximate setting a new definition of symmetric fair division is required. We say that an algorithm $\mathcal{F}$ gives an  $\varepsilon$-symmetric fair division when we have for all $i$ and all permutations $\sigma \in \mathfrak{S}_n$: 
$$\Big|\mu_i\big(\mathcal{F}(X,\underline{\mu},i) \big)- \mu_i\big(\mathcal{F}(X,\underline{\mu}^{\sigma},\sigma^{-1}(i)) \big)\Big| \leq \varepsilon.$$
This means that if we modify the order of the measures in the input of the algorithm then the perturbation on the new value obtained by the $i$-th player is bounded by~$\varepsilon$.\\

With these definitions it is natural to look for an $\varepsilon$-symmetric and $\varepsilon$-envy-free fair division. In this situation we do not need to repeat $n!$ times an $\varepsilon$-envy-free algorithm. Indeed, contrary to the exact setting there exists an algorithm computing an $\varepsilon$-perfect fair division, see \cite{Branzei}. This means that there exists an algorithm $\mathcal{F}$ such that 
$$\Big|\mu_i\big(\mathcal{F}(X,\underline{\mu},i)\big)-\dfrac{1}{n}\Big|\leq \varepsilon.$$
The complexity of this algorithm is in $O(n^2/\varepsilon)$.\\
Thus the $\varepsilon$-perfect algorithm gives  an $\varepsilon$-symmetric and $\varepsilon$-envy-free fair division  without increasing the complexity of an $\varepsilon$-envy-free protocol by a factor $n!$. Unfortunately, this algorithm  has an exponential time complexity in $n$ if we take into account the number of  elementary operations (arithmetic operations and inequality tests). Indeed, in this algorithm we have to consider all subsets $Y$ with cardinal at most $n(n-1)$ in a set with cardinal $nK$ where $K=\lceil \frac{2n(n-1)}{\varepsilon}\rceil$. Therefore, the asymptotic formula $\binom{2n}{n} \approx \dfrac{4^n}{\sqrt{\pi n}}$ shows that we have to consider an exponential number of subsets.
 \section{Aristotelian, symmetric and proportional cake cutting algorithms} \label{sec:propvrac}
 
In this section we first give an aristotelian and proportional fair division algorithm and then  a symmetric and proportional one. These two algorithms are based on Kuhn's algorithm, see \cite{Kuhn}.

\subsection{An aristotelian proportional cake cutting algorithm}\label{sec:aristote}
\subsubsection{The Evan-Paz algorithm and the last diminisher procedure are not aristotelian}
Before giving our aristotelian and proportional algorithm we show that the  classical Evan-Paz algorithm and the last diminisher procedure do not give an aristotelian fair division. \\

In the Evan-Paz algorithm we can have the following situation:
We consider four players with associated measures $\mu_1$, $\mu_2$, $\mu_3$, $\mu_4$. Furthermore, we suppose that $\mu_1=\mu_4$ is the Lebesgue measure on $[0,1]$. We also suppose that $\mu_2([0,0.5])=\mu_3([0,0.5])=1/2$ and $\mu_3([0.5,0.51])=1/4$.\\
 In the first step of the Evan-Paz algorithm we ask each player to cut the cake in two equal parts. More precisely, we ask $cut_i(0,1/2)$. In our situation, each player give the same point: $y=0.5$. In the second step, the algorithm consider two sets of two players. The first part of the cake $[0,0.5]$ will be given to the first set of players and the second part $[0.5,1]$ will be given to the second set of players. Usually, when all players give the same answers the two sets are constructed randomly or in function of the order of the players. Thus we can suppose that in the second step we give $[0,0.5]$ to $\mu_1$ and $\mu_2$ and $[0.5,1]$ to $\mu_3$ and $\mu_4$. At last, the ``Cut and Choose" algorithm is used to share $[0,0.5]$ (respectively $[0.5,1]$)  between the two players $\mu_1$, $\mu_2$ (respectively $\mu_3$, $\mu_4$). Thus $\mu_1$ cut the interval  $[0,0.5]$ and get $X_1$ such that $\mu_1(X_1)=1/4$, and $\mu_3$ cuts the interval $[0.5,1]$ and get $X_3=[0.5,0.51]$. Thus $X_4=[0.51,1]$ and $0.49=\mu_4(X_4)> \mu_1(X_1)=0.25$. As $\mu_1=\mu_4$, we deduce that the division is not aristotelian.\\

In the last diminisher procedure we can have the following situation: \\
We suppose that $\mu_1=\mu_2$ is the Lebesgue measure on $[0,1]$. Furthermore, we consider a measure $\mu_3$ such that $\mu_3([0,0.4])=1/3$, and  $\mu_3([1/3,0.5])=1/3$.\\
 In the first step of the last diminisher procedure we ask each player the query $cut_i(0,1/3)$. The first and second player give $\mu_1([0,1/3])=\mu_2([0,1/3])=1/3$ and the third player gives $\mu_3([0,0.4])=1/3$. In the first step of this algorithm we  give the portion $[0,1/3]$ to the first or to the second player. Suppose that we give this portion to the first player. In the second step of the last diminisher algorithm we ask $cut_2(1/3,1/3)$ and $cut_3(1/3,1/3)$. We  get thus the following information $\mu_2([1/3, 2/3])=1/3$ and $\mu_3([1/3, 0.5])=1/3$. After the second step  the algorithm  gives $[1/3,0.5]$ to the third player. It follows that the second player get $[0.5,1]$ and $\mu_2([0.5,1])=0.5>1/3=\mu_1([0,1/3])$. Therefore, this is not an aristotelian division since $\mu_1=\mu_2$. 

\subsubsection{An aristotelian proportional fair division algorithm}
 In this subsection, we recall Kuhn's fair division algorithm, see \cite{Kuhn}, and then we show how to modify it to get an aristotelian fair division algorithm. In order to state this algorithm we introduce the following definition:

\begin{Def}
Let $X=\sqcup_j A_j$ be a partition of $X$. An allocation relatively to this partition is a set $\{ (\mu_{i_1},A_{j_1}), \ldots, (\mu_{i_l},A_{j_l})\}$  such that for $k=1, \ldots , l$:  
$$\mu_{i_k}(A_{j_k})\geq \dfrac{\mu_{i_k}(X)}{n} \textrm{ and } \mu_{i}(A_{j_k})<\dfrac{\mu_i(X)}{n} \textrm{ if } i \neq i_1,\ldots, i_l.$$
A maximal allocation is an allocation whose cardinal is maximal.\\
In the following we say that a piece of cake $A_k$ is acceptable for the $i$-th player if $\mu_i(A_k)\geq \mu_i(X)/n$.
\end{Def}
In the previous definition the part $A_i$ is not necessarily given to the $i$-th player. The measurable sets $A_i$ do not play the same role than $X_i$ in the previous section. The partition $X=\sqcup_i A_i$ is just a partition of $X$, it is not necessarily the final result of a proportional fair division problem.\\

\begin{Lem}
For a given partition there always exists a maximal allocation.
\end{Lem}
\begin{proof}
With the  Frobenius-K\"onig theorem, Kuhn has shown in \cite{Kuhn} that there always exists an allocation relatively to a given partition.  This gives the existence of maximal allocations.
\end{proof}

Kuhn's algorithm proceeds as follows: The first player cuts the cake in $n$ parts with value $1/n=\mu_1(X)/n$ for his or her own measure. This gives a partition $X=\sqcup_i A_i$. Then we compute a maximal allocation relatively to this partition. Each player in the maximal allocation receives his or her associated portion. The remaining part of the cake is then divided between the rest of the players with the same method.\\

Now, we can describe our aristotelian algorithm. The idea is the following:\\
As before the first player cut the cake in $n$ parts with value $1/n$ for his or her own measure. This gives a partition $X=\sqcup_j A_j$ and we compute a maximal allocation relatively to this partition. Then each player $i_k$ in the maximal allocation receives his or her associated part if $\mu_{i_k}(A_{j})=1/n$ for all $A_j$ in the maximal allocation. In particular, all players with the same measure than the first player receive the same value. Then it remains two subcakes $X_1$ and $X_2$. We associate respectively these two subcakes to two set of players $\mathcal{E}_1$ and $\mathcal{E}_2$.\\

First, the set $\mathcal{E}_1$ corresponds to the set of players with an index $i_k$ in the maximal allocation such that there exists $A_j$ in the maximal allocation with \mbox{$\mu_{i_k}(A_j) \neq \mu_1(A_1)$}. Thus a player in $\mathcal{E}_1$ does not evaluate all portions $A_j$ as $\mu_1$. Then, we consider the set $\mathcal{L}_1$ constructed in the following way: $j_k \in \mathcal{L}_1$ if and only if $i_k \in \mathcal{E}_1$. At last, we set $X_1=\sqcup_{j \in \mathcal{L}_1}A_j$.
 
  Then we put together the players in $\mathcal{E}_1$  which seem to have the same measure. More precisely, we consider  a partition of $\mathcal{E}_1=\sqcup_{m=1}^d \mathcal{E}_{1,m}$ and $\mathcal{L}_1=\sqcup_{m=1}^d \mathcal{L}_{1,m}$ such that:
$$
(\star) \quad \begin{cases}
\forall i,i' \in \mathcal{E}_{1,m},\, \forall j, \quad \mu_i(A_j)=\mu_{i'}(A_{j}),\\
\mathcal{L}_{1,m}=\{j_k \, | \, i_k \in \mathcal{E}_{1,m} \}.\\

\end{cases}$$
This means that for all $i \in \mathcal{E}_{1,m}$, there exists  a constant $c_{j,m}$ (independent of $i$) such that for all $j$ we have $\mu_i(A_j)=c_{j,m}$.\\
In particular, as $\mu_{i_k}(A_{j_k}) \geq \mu_{i_k}(X)/n$, we have the following
\begin{Rem}\label{rem:tech}
For all $i \in \mathcal{E}_{1,m}$ and $j \in \mathcal{L}_{1,m}$ we have $\mu_i(A_j)\geq \mu_i(X)/n$.
\end{Rem}

Then we consider $X_{1,m}=\sqcup_{j \in \mathcal{L}_{1,m}}A_j$ and we associate to these subcakes the players with indices in $\mathcal{E}_{1,m}$. Therefore, it will be possible to share $X_{1,m}$ between the players with indices in $\mathcal{E}_{1,m}$ because by construction they evaluate all $A_j$ in the same way with a value bigger than $1/n$.\\

At last, we  denote by $X_2$ the part of the cake not appearing in the maximal allocation. Then we can share $X_2$ between the players not appearing in the maximal allocation since by definition they do not find acceptable the portions  in the maximal allocation.\\
The  algorithm will call recursively the algorithm on $X_{1,m}$ and $X_2$.\\

In the following we will  use queries for a ``subcake" $\mathcal{X} \subsetneq [0,1]$. Indeed, as in Kuhn's algorithm we are going to consider situations where the cake will be of the form $[0,1] \setminus Y$, where $Y$ will correspond to the part of the cake already given by the algorithm. We need thus the following notations:
\begin{enumerate}
\item $eval_i^{\mathcal{X}}(x,y)$: Ask agent $i$ to evaluate $[x,y]\cap \mathcal{X}$.\\ This means return $\mu_i([x,y]\cap \mathcal{X})$.
\item $cut_i^{\mathcal{X}}(x,a)$: Ask agent $i$ to give $y$ such that $\mu_i([x,y]\cap \mathcal{X})=a$.
\end{enumerate}
We will see that these queries do not introduce new operations. More precisely, during the algorithm these queries  $eval_i^{\mathcal{X}}(x,y)$ and $cut_i^{\mathcal{X}}(x,a)$ can be compute thanks to  $eval_i(x,y)$ and $cut_i(x,a)$.\\

\texttt{AristoProp}\\
\textsf{Inputs:} $\underline{\mu}=[\mu_1,\ldots,\mu_{n}]$, $\mathcal{X} \subset [0;1]$.\\
\textsf{Outputs:} $\mathcal{X}=\mathcal{F}(\mathcal{X},\underline{\mu},1) \sqcup \ldots \sqcup \mathcal{F}(\mathcal{X},\underline{\mu},n)$, where $\mathcal{F}(\mathcal{X},\underline{\mu},i)$ is a finite union of disjoint intervals and $\mathcal{F}(\mathcal{X},\underline{\mu},i)$ is given to the $i$-th player.

\begin{enumerate}
\item \label{step1:aristo}\%\emph{Ask the first player to cut the cake in $n$ parts with values $\mu_1(\mathcal{X})/n$. }\%\\
\%\emph{This gives: $\mathcal{X}=\sqcup_i A_i$.}\% \\
$x_0:=\min_{x \in \mathcal{X}}(x)$\\
For $j$ from 1 to $n$ do\\
\hphantom{bla} $x_j:=cut_1^{\mathcal{X}}\big(x_{j-1},\mu_1(\mathcal{X})/n\big),$\\
\hphantom{bla} Set $A_j:=[x_{j-1};x_{j}]\cap \mathcal{X}$.\\

\item \label{step2:aristo} \% \emph{Ask each player to evaluate each $A_j$.}\%\\
For $i$ from 2 to $n$ do  \\
\hphantom{bla} For $j$ from $1$ to $n$ do\\
\hphantom{blabla} $eval_i^{\mathcal{X}}(x_{j-1},x_j)$.\\

\item \label{step3:aristo} Compute a maximal allocation $\mathcal{A}:=\{ (\mu_{i_1},A_{j_1}), \ldots, (\mu_{i_l},A_{j_l})\}$  relatively to the partition $\mathcal{X}=\sqcup_i A_i$.\\

\item \label{step4:aristo}\% \emph{If for  all $j$ in $\{j_1, \ldots, j_l\}$, we have $\mu_{i_k}(A_{j})=\mu_1(A_1)$ then give the portion $A_{j_k}$ to the player with associated measure $\mu_{i_k}$.}\%\\
Set $\mathcal{E}:=\emptyset$, $\mathcal{E}_1:=\emptyset$, $\mathcal{L}_1:=\emptyset$, $\mathcal{X}_1:=\emptyset$.\\
For $i_k$ in $\{ i_1, \ldots, i_l\}$ do \\
\hphantom{bla}  t:=true;\\
\hphantom{bla} For $j$ in $\{ j_1, \ldots, j_l\}$ do\\
\hphantom{blabla} If $\mu_{i_k}(A_{j}) \neq \mu_1(A_1)$ Then t:=false.\\
\hphantom{bla} If t=true Then $\mathcal{F}(X,\underline{\mu},i_k):=A_{j_k}$, $\mathcal{E}:=\mathcal{E} \cup \{i_k\}$,\\
\hphantom{bla If t=true} Else $\mathcal{E}_1:=\mathcal{E}_1 \cup \{i_k\}$, $\mathcal{L}_1:=\mathcal{L}_1 \cup \{j_k \}$, $\mathcal{X}_1:=\mathcal{X}_1 \cup A_{j_k}$.\\

\item \label{step5:aristo} 
Construct a partition $\mathcal{E}_1=\sqcup_{m=1}^d \mathcal{E}_{1,m}$ and a partition $\mathcal{L}_1:=\sqcup_{m=1}^d \mathcal{L}_{1,m}$ sastisfying $(\star)$.\\

\noindent Set $\underline{\mu}_{1,m}$ as the list of measures associated to players with index  in $\mathcal{E}_{1,m}$. \\
Set $\mathcal{X}_{1,m}:=\sqcup_{j \in \mathcal{L}_{1,m}} A_{j}$.\\
Set $\mathcal{E}_2:=\{1,\ldots,n\} \setminus \{i_1, \ldots,i_l\}$, $\mathcal{X}_2:=\mathcal{X} \setminus \big(\sqcup_{k=1}^l A_{j_k}\big)$.\\
Set $\underline{\mu}_2$ as the list of measures associated to players with index  in $\mathcal{E}_2$. \\

\item \label{step6:aristo} Return\big($\sqcup_{i \in \mathcal{E}} \mathcal{F}(\mathcal{X},\underline{\mu},i) \sqcup_{m=1}^d$ \texttt{AristoProp} $(\underline{\mu}_{1,m},\mathcal{X}_{1,m}) \sqcup$ \texttt{AristoProp} $(\underline{\mu}_2,\mathcal{X}_2)\big)$.

\end{enumerate}

\begin{Prop}\label{prop:aristo_prop}
The algorithm \texttt{AristoProp} applied to $\underline{\mu}=[\mu_1,\ldots,\mu_n]$ and  \mbox{$\mathcal{X}=[0,1]$} terminates and  is aristotelian. 
\end{Prop}

\begin{proof}
The algorithm terminates since after one call of the algorithm the number of player decreases strictly since the first player always get a part of the cake.\\

Now, we are going to prove by induction that the algorithm is aristotelian.\\
We consider the following claim:\\
$(H_n)$: The algorithm \texttt{AristoProp} applied with $n$ measures is aristotelian.\\

For $n=2$, $H_2$ is true. Indeed, if $\mu_1=\mu_2$ then  $\mu_2$ belongs to the maximal allocation computed in Step~\ref{step3:aristo}. Furthermore, we can suppose without loss of generality that the maximal allocation has the following form $\mathcal{A}=\{(\mu_1, A_1),(\mu_2,A_2)\}$.\\
By construction we have $\mu_1(A_1)=\mu_1(A_2)$. Thus $\mu_1(A_1)=\mu_2(A_1)=\mu_2(A_2)$ since $\mu_1=\mu_2$. As in Step~\ref{step4:aristo}, $\mu_1$ gets the portion $A_1$ and $\mu_2$ gets the portion $A_2$, we deduce that $H_2$ is true.\\

Now, we suppose that $H_k$ is true when $k \leq n$ and we are going to prove that $H_{n+1}$ is true.\\

We suppose that we have $n+1$ measures $\mu_1, \ldots, \mu_{n+1}$, and that $\mu_{p}=\mu_{q}$, where $p,q \in\{1,\ldots,n+1\}$.\\

First, we remark that if $(\mu_{p},A_{j_p})$ belongs to a maximal allocation then $\mu_{q}$ also belongs to the same maximal allocation. Indeed, if $\mu_q$ does not belong to the maximal allocation then $\mu_q(A_{j_p}) <\mu_q(\mathcal{X})/n$ but $\mu_q(A_{j_p})=\mu_p(A_{j_p})\geq \mu_p(\mathcal{X})/n$ because $\mu_p=\mu_q$ and $(\mu_p,A_{j_p})$ belongs to the maximal allocation. This gives the desired contradiction and proves our remark.\\

Now two situations appear: In Step~\ref{step3:aristo}, $\mu_p$ and $\mu_q$ belongs to the maximal allocation or they do not belong to it.\\

If $\mu_p$ and $\mu_q$ do not belong to the maximal allocation then $\mu_p$ and $\mu_q$ belong to the list  $\underline{\mu}_2$. Then, $\mu_p$ and $\mu_q$ get their portions when, in Step~\ref{step6:aristo}, we apply \texttt{AristoProp} $(\underline{\mu}_2,\mathcal{X}_2)$. As the list $\underline{\mu}_2$ have at most $n$ measures and $H_n$ is true we deduce that $\mu_p$ and $\mu_q$ get the same value and then the algorithm is aristotelian in this case.\\

If $\mu_p$ and $\mu_q$ belong to the maximal allocation then we have two cases:\\
there exists an index $j_0$ in $\{j_1, \ldots, j_l\}$ such that $\mu_p(A_{j_0})\neq \mu_1(A_1)$ or\\
for all $j_k$ in  $\{j_1, \ldots, j_l\}$ we have $\mu_p(A_{j_k})=\mu_1(A_1)$.\\

In the first case, as $\mu_p=\mu_q$ then we also have $\mu_q(A_{j_0})\neq \mu_1(A_1)$. Then, $p$ and $q$ belong to  $\mathcal{E}_1$. Furthermore, as $\mu_p=\mu_q$, for all $j$ we have $\mu_p(A_j)=\mu_q(A_j)$. Then $\mu_p$ and $\mu_q$ belong to the same list $\underline{\mu}_{1,m}$. As the list $\underline{\mu}_{1,m}$ have at most $n$ measures and $H_n$ is true we deduce that $\mu_p$ and $\mu_q$ get the same value and then the algorithm is also aristotelian in this case.\\

In the second case,  we have $p,q \in \mathcal{E}$ and $\mu_p(A_{j_k})=\mu_q(A_{j_k})=\mu_1(A_1)$ for all $j_k$ in  $\{j_1, \ldots, j_l\}$. Thus the $p$-th and $q$-th player  evaluate in the same way the portion they get. Thus, the algorithm is also aristotelian in this case and this concludes the proof.
\end{proof}

\begin{Prop}\label{prop:aristo_prop2}
The algorithm \texttt{AristoProp} applied to $\underline{\mu}=[\mu_1,\ldots,\mu_n]$  gives a proportional fair division of $[0,1]$.
\end{Prop}

\begin{proof}
We are going to prove this result by induction. We consider the following claim:\\
$(H_n)$: The algorithm \texttt{AristoProp} applied to $n$ measures  gives a proportional fair division of $[0,1]$.\\

For $n=1$, $H_1$ is true.

Now, we suppose that $H_k$ is true for $k \leq n$ and we are going to prove that $H_{n+1}$ is true.\\

We consider $n+1$ measures $\mu_1$, \ldots, $\mu_{n+1}$.\\
If $i \in \mathcal{E}$ then the $i$-th player receive a portion $\mathcal{F}(X,\underline{\mu}, i)$ that he or she consider to have a value equal to $\mu_1(A_1)$. As $\mu_1(A_1)=\mu_1(X)/n$, we get $\mu_i\big( \mathcal{F}(X,\underline{\mu}, i) \big) \geq \mu_1(X)/n$. Thus, the algorithm is proportional in this case.\\

If $i \not \in \mathcal{E}$ then we have the following situation: $\mu_i$ belongs to a list $\underline{\mu}_{1,m}$ or to the list $\underline{\mu}_2$.\\

If $\mu_i$ belongs to the list $\underline{\mu}_2$, then $\mu_i$ does belong to the maximal allocation considered and for all $j_k \in \{j_1, \ldots,j_l\}$ we have $\mu_i(A_{j_k})<\mu_1(A_1)=\mu_i(X)/n$. Therefore, we have
$$(\sharp) \quad \mu_i(\mathcal{X}_2)=\mu_i(X)-\sum_{k=1}^l \mu_i(A_{j_k})\geq \mu_i(X)-\dfrac{l\mu_i(X)}{n}=\dfrac{(n-l) \mu_i(X)}{n},$$
where $l$ is the size of the maximal allocation computed in Step~\ref{step3:aristo}.\\

Thanks to our induction hypothesis, in Step~\ref{step6:aristo} the $i$-th player receives at least \mbox{$\mu_i(\mathcal{X}_2)/(n-l)$} for his or her own measure. Thus, by $(\sharp)$ the $i$-th player gets at least $\mu_i(X)/n$. Thus, in this case, the algorithm is proportional. \\

If $\mu_i$ belongs to a list $\underline{\mu}_{1,m}$, by Remark~\ref{rem:tech}, we have $\mu_i(A_j) \geq \mu_i(X)/n$, for all $j \in \mathcal{L}_{1,m}$. Thus
$$(\sharp \sharp) \quad \mu_i(\mathcal{X}_{1,m})=\mu_i(\sqcup_{j \in\mathcal{L}_{1,m}} A_j) \geq \dfrac{|\mathcal{L}_{1,m}|\mu_i(X)}{n},$$
where $|\mathcal{L}_{1,m}|$ is the number of elements in $\mathcal{L}_{1,m}$ and this number is equal to $\mathcal{E}_{1,m}$ the number of measures in the list $\underline{\mu}_{1,m}$.\\

Thanks to our induction hypothesis, in Step~\ref{step6:aristo} the $i$-th player receives at least $\mu_i(\mathcal{X}_{1,m})/|\mathcal{E}_{1,m}|$ for his or her own measure. Thus, by $(\sharp \sharp)$ the $i$-th player gets at least $\mu_i(X)/n$ and this concludes the proof.
\end{proof}

\begin{Prop}\label{prop:complexityaristot}
The algorithm \texttt{AristoProp} applied to $\underline{\mu}=[\mu_1,\ldots,\mu_n]$, and $X=[0,1]$ uses at most $\bigO(n^3)$ queries in the Robertson-Webb model.
\end{Prop}
\begin{proof}
In Step \ref{step1:aristo} we use $n$  $cut_i$ queries. In Step \ref{step2:aristo} we use $n(n-1)$,  $eval_i$ queries.\\
During the first call of the algorithm we use the $cut_i$ and $eval_i$ queries. In the next calls the algorithm uses the $cut_i^{\mathcal{X}}$ and $eval_i^{\mathcal{X}}$ queries where $\mathcal{X}=\sqcup_j A_j$  and the measures of $A_j$ are known by the players thanks to Step \ref{step2:aristo} of the algorithm. We remark that the situation $\mathcal{X}_2=\mathcal{X}\setminus \Big( \sqcup_{k=1}^l A_{j_k} \Big)$ of Step \ref{step5:aristo} corresponds to $\mathcal{X}_2=\sqcup_{j \not \in \mathcal{L}\sqcup \mathcal{L}_1} A_j$. 
It follows that we can write $\mathcal{X}$ in the following form: $\mathcal{X}=\sqcup_{j=1}^k [s_j;t_j]$, where $s_1<t_1<s_2<t_2<\cdots<s_k<t_k$ and the measures of $[s_j;t_j]$ and $[t_j;s_{j+1}]$ are known thanks to the previous calls of the algorithm.\\
Now, we define a function $\mathsf{f}$ in order to explain how we compute $eval_i^{\mathcal{X}}$ from $eval_i$: If $s_{j_0}<x< t_{j_0}$ then we set $\mathsf{f}(x)=j_0$.\\

If $\mathsf{f}(x)=\mathsf{f}(y)$ then $[x;y]\subset [s_{j_0};t_{j_0}]$ and $eval_i^{\mathcal{X}}(x,y)=eval_i(x,y)$,\\
 else we have 
 \begin{eqnarray*}
 eval_i^{\mathcal{X}}(x,y)&=&\mu_i\big([x,y]\cap \mathcal{X}\big)=\mu_i\Big([x,y] \cap \big( \sqcup_{j=1}^k [s_j,t_j]\big) \Big)\\
 &=&eval_i(x,y)-\sum_{j=\mathsf{f}(x)}^{\mathsf{f}(y)-1} eval_i(t_{j},s_{j+1}).
\end{eqnarray*}
As  $\mu_i([t_{j},s_{j+1}])=eval_i(t_{j},s_{j+1})$ is known thanks to the previous calls of the algorithm, the query  $eval_i^{\mathcal{X}}(x,y)$ needs just one new query: $eval_i(x,y)$.\\

For the $cut_i^{\mathcal{X}}$ query we proceed in the following way:\\
Suppose that we want to compute $cut_i^{\mathcal{X}}(x,a)$.\\
First, compute $eval_i(x,t_{\mathsf{f}(x)})$. As we know $\mu_i([s_j;t_j])$ for $j=1,\ldots,k$ then with all these values we can deduce in which interval $[s_{1},t_1]$,\ldots, $[s_{k},t_{k}]$ is the cutpoint $y$. We denote by $[\alpha,\beta]$ this interval. Thanks to the knowledge of $\mu_i([s_j;t_j])$ and $\mu_i([x,t_{\mathsf{f}(x)}])$ we can also get $a'=eval_i^{\mathcal{X}}(x,\alpha)$. Then we have:\\ $cut_i^{\mathcal{X}}(x,a)=cut_i(\alpha,a-a')$.\\
Therefore, $cut_i^{\mathcal{X}}$ needs two new queries in the Robertson-Webb model: $eval_i(x,t_{\mathsf{f}(x)})$ and $cut_i(\alpha,a-a')$.\\

In conclusion, in Step \ref{step1:aristo} the algorithm applied with $\eta$ measures uses $\eta$ $cut_i^{\mathcal{X}}$ queries, thus these queries can be computed with $2 \eta$  queries in the Robertson-Webb model.  In Step \ref{step2:aristo}, it uses  $\eta(\eta-1)$ $eval_i^{\mathcal{X}}$ queries. These queries can be computed with $\eta(\eta-1)$ queries in the Robertson-Webb model. Therefore, each call of the algorithm applied with $\eta$ measures uses $\eta(\eta+1)$ queries in the Robertson-Webb model of computation. Furthermore, in the worst case, at each call of the algorithm only one player get a part of the cake. Thus we use at most 
$$n^2+\sum_{\eta=1}^{n-1} \eta(\eta+1)=n^2+ \sum_{\eta=1}^{n-1}\eta^2 +\sum_{\eta=1}^{n-1} \eta=n^2+ \dfrac{n(n-1)(2n-1)}{6}+\dfrac{n(n-1)}{2} \in \bigO(n^3)$$ queries in the Robertson-Webb model.
\end{proof}

From the previous propositions we get:
\begin{Thm}
There exists an aristotelian proportional fair division algorithm which uses at most $\bigO(n^3)$ queries in the Robertson-Webb model of computation.
\end{Thm}

As already mentioned in the introduction, this theorem says that if we just want an aristotelian proportional fair division it is not necessary to use an envy-free algorithm which uses an exponential number of queries. 

 
\subsection{A symmetric and proportional cake cutting algorithm}
 In Section~\ref{sec:symenvyfree}, we have proposed a symmetric and envy-free protocol, this gives then a proportional and symmetric protocol. With this approach we need to compute $n!$ envy-free fair divisions. This raises the following question: Do we need to compute  a proportional fair division for all the possible permutations to get a proportional and symmetric division algorithm?   \\
  In this subsection we are going to show that there exists a symmetric and proportional algorithm which uses $\bigO(n^3)$ queries in the Robertson-Webb model.\\

The idea of the algorithm is a kind of improvement of the aristotelian algorithm. Indeed, in the aristotelian algorithm if two players get a portion at the same stage of the algorithm then they will evaluate their portion in the same way. Here, we construct an algorithm in order to have also the following property: a player will always receive a portion at the same stage of the algorithm whatever  his or her position in the input list $\underline{\mu}$ is.\\

Our algorithm works as follows: Instead of asking to the first player to divide the cake in $n$ equal parts, we are going to ask to all players to cut the cake in $n$ equal parts for their own measures. Then, we will select the ``smallest partition" relatively to the graded order. Thus, we obtain a partition $X=\sqcup_{j}A_j$ independent of the order of the measures.\\
Next, we compute all maximal allocations $\mathcal{A}$ relatively to this partition. For all of these allocations we consider the set $\mathcal{E}_{\mathcal{A}}$, constructed as follows: $i \in \mathcal{E}_{\mathcal{A}}$ if and only if $i$ belongs to the maximal allocation $\mathcal{A}$ and for all $j$, we have: $\mu_i(A_j)=\mu_i(X)/n$. These sets will play the same role as the set $\mathcal{E}$ in the algorithm \texttt{AristoProp}. However, here we have several sets $\mathcal{E}_{\mathcal{A}}$ and then we have to choose one of them.
 We select then a maximal allocation $\mathcal{A}$ where the portions associated to the players in $\mathcal{E}_{\mathcal{A}}$ appear in the leftmost part of the cake. Thus this choice is still independent of the order or the players. At last, we give these portions to their associated players. If a portion $A_j$ belongs to the selected maximal allocation is not given to a player then this portion is used to construct a subcake $\mathcal{X}_1$ as in the aristotelian case. If a portion $A_j$ is not in the selected allocation then this portion is used to construct the subcake $\mathcal{X}_2$.
Our algorithm is then constructed in a way such that we  always associate the same players to the same subcake $\mathcal{X}_1$ or $\mathcal{X}_2$. As we repeat  our strategy on $\mathcal{X}_1$ and $\mathcal{X}_2$ we get a symmetric algorithm. \\


\texttt{SymProp}\\
\textsf{Inputs:} $\underline{\mu}=[\mu_1,\ldots,\mu_{n}]$, $\mathcal{X} \subset [0,1]$.\\
\textsf{Outputs:} $\mathcal{X}=\mathcal{F}(\mathcal{X},\underline{\mu},1) \sqcup \ldots \sqcup \mathcal{F}(\mathcal{X},\underline{\mu},n)$, where $\mathcal{F}(\mathcal{X},\underline{\mu},i)$ is a finite union of disjoint intervals and $\mathcal{F}(\mathcal{X},\underline{\mu},i)$ is given to the $i$-th player.

\begin{enumerate}
\item \label{step1:sym}\%\emph{Ask all players to cut the cake in $n$ parts with values $\mu_i(\mathcal{X})/n$.} \%\\
For $i$ from $1$ to $n$ do\\
\hphantom{bla} $x_{i,0}:=\min_{x \in \mathcal{X}}(x)$\\
\hphantom{bla} For $j$ from 1 to $n$ do\\
\hphantom{blabla} $x_{i,j}:=cut_i^{\mathcal{X}}(x_{i,j-1},\mu_i(\mathcal{X})/n)$.\\

\item \label{step2:sym}\%\emph{ Find the smallest partition $\mathcal{X}=\sqcup_{j=1}^n A_j$ for the graded order. }\%\\
Compute $(x_{0,0}, \ldots, x_{0,n}):=\min_{\succ_{gr}} \{ (x_{i,0},\ldots,x_{i,n}) \, | \, i=1,\ldots,n \}$.\\
For $j$ from $1$ to $n$ do\\
\hphantom{bla} Set $A_j:=[x_{0,j-1};x_{0,j}] \cap \mathcal{X}$.\\

\item \label{step3:sym} \% \emph{Ask each player to evaluate each $A_j$.}\%\\
For $i$ from $1$ to $n$ do  \\
\hphantom{bla}  For $j$ from $1$ to $n$ do\\
\hphantom{blabla} $eval_i^{\mathcal{X}}(x_{0,j-1},x_{0,j})$.\\

\item  \label{step4:sym} Compute the set $S$ of all maximal allocations $\mathcal{A}:=\{ (\mu_{i_1},A_{j_1}), \ldots, (\mu_{i_l},A_{j_l})\}$  relatively to the partition $\mathcal{X}=\sqcup_{j=1}^n A_j$.\\

\item \label{step5:sym} \% \emph{The set $\mathcal{E}_{\mathcal{A}}$ is the  set of indices  $i \in \{i_1, \ldots,i_l\}$   appearing in the allocation $\mathcal{A}$ such that for $j=1, \ldots, n$, we have $\mu_i(A_j)=\mu_i(\mathcal{X})/n$.}\%\\
For all  $\mathcal{A}=\{ (\mu_{i_{1}},A_{j_{1}}), \ldots, (\mu_{i_{l}},A_{j_{l}}) \}$ in $S$ do\\
Set $\mathcal{E}_{\mathcal{A}}:=\emptyset$.\\
\hphantom{bla} For $k$ from 1 to $l$ do\\
\hphantom{blablabla} If the vector $(x_{i_k,0}, \ldots,x_{i_k,n})$ associated to $\mu_{i_{k}}$ satisfied \\
\hphantom{blablabla} $(x_{i_k,0}, \ldots,x_{i_k,n})=(x_{0,0},\ldots,x_{0,\eta})$\\
\hphantom{blablabla} Then $\mathcal{E}_{\mathcal{A}}:=\mathcal{E}_{\mathcal{A}} \cup \{i_k\}$.\\
\hphantom{bla}

 \item \label{step6:sym} \%\emph{Find an allocation $\hat{\mathcal{A}} \in S$ such that the portions associated to the measures $\mu_{i}$ with $i \in \mathcal{E}_{\hat{\mathcal{A}}}$ are on the leftmost part of $\mathcal{X}$.}\%\\
 For all allocations $\mathcal{A}=\{ (\mu_{i_{1}},A_{j_{1}}), \ldots, (\mu_{i_{l}},A_{j_{l}}) \}$ in $S$ do\\
\hphantom{bla} $N_{\mathcal{A}}:=0$;\\
\hphantom{bla} For $k$ from 1 to $l$ do \\
 \hphantom{blabla} If  $i_k \in \mathcal{E}_{\mathcal{A}}$ Then $N_{\mathcal{A}}:=N_{\mathcal{A}}+2^{j_{k}}$.\\

\noindent Find an allocation $\hat{\mathcal{A}}\in S$ such that $N_{\hat{\mathcal{A}}}$ is minimal.\\

\item \label{step7:sym} \% \emph{We consider the allocation $\hat{\mathcal{A}}$. If $i_k \in \mathcal{E}_{\hat{\mathcal{A}}}$ then we give the portion $A_{j_k}$ to the $i_k$-th player else we use $A_{j_k}$ to construct the subcake $\mathcal{X}_1$.} \%\\
Set $\hat{\mathcal{A}}:=\{(\mu_{i_1},A_{j_1}),\ldots,(\mu_{i_l},A_{j_l})\}$.\\
Set $N_{\hat{\mathcal{A}}}:=\sum_{j \in J} 2^j$,\\
Set $\mathcal{E}_1:=\emptyset$, $\mathcal{L}_1:=\emptyset$, $\mathcal{X}_1:=\emptyset$.\\
For $j_k$ in $\{ j_1, \ldots, j_l\}$ do \\
\hphantom{bla} If $j_k \in J$ Then $\mathcal{F}(X,\underline{\mu},i_k):=A_{j_k}$\\
\hphantom{bla If $i_k \in\mathcal{E}_{E}$} Else $\mathcal{E}_1:=\mathcal{E}_1 \cup \{i_k\}$, $\mathcal{L}_1:=\mathcal{L}_1 \cup \{j_k \}$, $\mathcal{X}_1:=\mathcal{X}_1 \cup A_{j_k}$.\\


\item \label{step8:sym} 
Construct a partition $\mathcal{E}_1=\sqcup_{m=1}^d \mathcal{E}_{1,m}$ and a partition $\mathcal{L}_1:=\sqcup_{m=1}^d \mathcal{L}_{1,m}$ sastisfying $(\star)$.\\
Set $\underline{\mu}_{1,m}$ as the list of measures associated to players with index  in $\mathcal{E}_{1,m}$. \\
Set $\mathcal{X}_{1,m}:=\sqcup_{j \in \mathcal{L}_{1,m}} A_{j}$.\\
Set $\mathcal{E}_2:=\{1,\ldots,n\} \setminus \{i_1, \ldots,i_l\}$, $\mathcal{X}_2:=\mathcal{X} \setminus \big(\sqcup_{k=1}^l A_{j_k}\big)$.\\
Set $\underline{\mu}_2$ as the list of measures associated to players with index  in $\mathcal{E}_2$. \\

\item \label{step9:sym} Return\big($\sqcup_{i \in \mathcal{E}} \mathcal{F}(\mathcal{X},\underline{\mu},i) \,\sqcup_{m=1}^d$ \texttt{SymProp} $(\underline{\mu}_{1,m},\mathcal{X}_{1,m})\, \sqcup$ \texttt{SymProp} $(\underline{\mu}_2,\mathcal{X}_2)\big)$.
\end{enumerate}

\begin{Prop}\label{prop:sym_prop}
The algorithm \texttt{SymProp} applied to $\underline{\mu}=[\mu_1,\ldots,\mu_n]$ and  $\mathcal{X}=[0,1]$ terminates,  is symmetric and gives a proportional fair division of $[0,1]$.
\end{Prop}

\begin{proof}
The algorithm terminates since after one call of the algorithm the number of player decreases strictly since at least one player get a part of the cake.\\

Now, we  have to prove that this algorithm is symmetric.\\
First, we remark that in Step~\ref{step2:sym} the partition $\mathcal{X}=\sqcup_j A_j$ is independent of the order of the players given in the input. Indeed, this partition is chosen thanks to the graded order.\\

Second, we define the set $\mathbb{A}$ as the set of portions $A_{j_k}$ given in Step~\ref{step7:sym} and the set $\mathbb{M}$ as the set of measures receiving a portion in Step~\ref{step7:sym}.\\
We are going to show that these two sets are independent of the order of the players and also independent of the choice of $\hat{\mathcal{A}}$ in Step~\ref{step6:sym}.\\

 Indeed, we give $A_{j_k}$ if $j_k \in J$, where $J$ is the set defined by the property $N_{\hat{\mathcal{A}}}=\sum_{j \in J} 2^j$. As the binary expansion of $N_{\hat{\mathcal{A}}}$ is unique, the set $J$ is independent of the order of the players and also independent of the allocation $\hat{\mathcal{A}}$ chosen in Step~\ref{step6:sym}. Thus $\mathbb{A}$ is independent of the order of the players and is also independent of the choice of $\hat{\mathcal{A}}$.\\
The set  $\mathbb{M}$ is  independent of the order of the players and of the choice of $\hat{\mathcal{A}}$ in Step~\ref{step6:sym}.\\
 Indeed, if  $\mu_{i_k} \in \mathbb{M}$ then $(\mu_{i_k},A_{j_k}) \in \hat{\mathcal{A}}$ and $j_k \in J$, where $N_{\hat{\mathcal{A}}}=\sum_{j \in J}2^j$. Therefore as before $\mathbb{M}$ is independent of the choice of $\hat{\mathcal{A}}$ in Step~\ref{step6:sym}. Furthermore, by contruction in Step~\ref{step5:sym} and Step~\ref{step6:sym}, $i_k$ belongs to $\mathcal{E}_{\hat{\mathcal{A}}}$. Thus, $\mu_{i_k} \in \mathbb{M}$ means that the associated  vector $(x_{i_k,0}, \ldots,x_{i_k,n})$ satisfies the equality $(x_{i_k,0}, \ldots,x_{i_k,n})=(x_{0,0},\ldots,x_{0,n})$. As the choice of the partition $\mathcal{X}=\sqcup_j A_j$ in Step~\ref{step2:sym} is independent of the order of the players, that is to say the choice of  $(x_{0,0},\ldots,x_{0,n})$ is independent of the order of the players, we deduce that $\mathbb{M}$ is indepedent of the order of the players. \\

Now, we remark that, for all $\mu \in \mathbb{M}$ and all $A \in \mathbb{A}$ we have: $\mu(A)=\mu(\mathcal{X})/n$.\\
Indeed, all players $i_k$ associated to a measure $\mu_{i_k} \in \mathbb{M}$ belongs to $\mathcal{E}_{\hat{A}}$ by construction in Step~\ref{step5:sym} and Step~\ref{step6:sym}. Thus $\mu(A)$ is independent of the choice of $\hat{\mathcal{A}}$  in Step~\ref{step6:sym}. \\

Therefore, in Step~\ref{step7:sym} the sets $\mathbb{A}$ and $\mathbb{M}$ and the value $\mu(A)$ for $\mu \in \mathbb{M}$ and $A \in \mathbb{A}$ are independent of the order of the players and independent of the choice of $\hat{\mathcal{A}}$ with  minimal $N_{\hat{\mathcal{A}}}$ in Step~\ref{step6:sym}. Furthermore, we can deduce that:
\begin{itemize}
\item[-] $\mathcal{X}_1$ and $\mathcal{E}_1$ and then $\mathcal{X}_{1,m}$ and $\mathcal{E}_{1,m}$,
\item[-] $\mathcal{X}_2$ and $\mathcal{E}_2$
\end{itemize}
are also independent of the order of the players and of the choice of $\hat{A}$ in Step~\ref{step6:sym}. It follows then that the algorithm is symmetric.\\

The algorithm is proportional.\\
Indeed, the sets $\mathcal{X}_1$ and $\mathcal{X}_2$ are  constructed as in \texttt{AristoProp}. The strategy used by \texttt{SymProp}  is the same than the one used in the algorithm \texttt{AristoProp}.
Thus with the same approach as the one used in Proposition \ref{prop:aristo_prop}  we can deduce that the algorithm \texttt{SymProp} is proportional.
\end{proof}

\begin{Prop}
The algorithm \texttt{SymProp} uses at most $\bigO(n^3)$ queries in the Robertson-Webb model.
\end{Prop}

\begin{proof}
In Step \ref{step1:sym} we use $n^2$ $cut_i^{\mathcal{X}}$ queries, in Step \ref{step3:sym} we use $n(n-1)$ $eval_i^{\mathcal{X}}$ queries. Thus as shown in Proposition \ref{prop:complexityaristot} we use at most $\bigO(n^3)$ queries in the Robertson-Webb model.
\end{proof}



\begin{Rem}
Suppose that all the measures $\mu_i$ are equal to the Lebesgue measure on $[0,1]$. Then in Step \ref{step2:sym} of the algorithm we have 
$$(x_{0,0},\ldots,x_{0,n})=\big(0,1/n,2/n,\ldots,(n-1)/n,1\big),$$
and then $S$ contains $n!$ allocations.\\
Thus in Step \ref{step6:sym} we compute $n!$ times the number $N_{\mathcal{A}}=2+2^2+\cdots+2^n$.\\
Therefore, there exists a situation where the algorithm computes at least $n!$ sums.\\

This is not the only situation where we need to perform an exponential number of arithmetic operations. Another example is the following: Consider $2n+1$ players, suppose that the measure associated to the first $n$ players is the Lebesgue measure on $[0;1]$ and the measure associated to the other players is concentrated on $[\frac{2n}{2n+1},1]$. Then in Step \ref{step2:sym} of \texttt{SymProp} we have 
$$(x_{0,0},\ldots,x_{0,n})=\Big(0,\frac{1}{2n+1},\ldots,\frac{2n}{2n+1},1\Big),$$
and $S$ contains $\binom{2n}{n}$ allocations.
Indeed, in order to get a maximal allocation we have to associated  $n$ intervals among the $2n$ first intervals to the $n$ players with the Lebesgue measure on $[0,1]$. As $\binom{2n}{n} \approx \dfrac{4^n}{\sqrt{\pi n}}$ in this situation we also perform an exponential number of  operations or inequality tests.\\
\end{Rem}
These examples show that the algorithm \texttt{SymProp} needs  a polynomial number of queries in the Robertson-Webb model but needs an exponential number of elementary operations. The combinatorial nature of the problem is processed with classical arithmetic operations and inequality tests.
\section{Conclusion}\label{sec:conclusion}
In this article we have given an algorithm for computing symmetric and envy-free fair division.\\
The complexity in the Robertson-Webb model  of this algorithm increases the complexity of an envy-free fair division by a factor $n!$. This raises the following question: Can we avoid or reduce this factor? \\
Furthermore, we know a lower bound for the number of queries for an envy-free division. This lower bound is $\Omega(n^2)$, see \cite{Procaccia-lowerbound}.  What is the lower bound for the symmetric and envy-free problem? Do we have necessarily a factorial number of queries? In other words, does the lower bound for the symmetric and envy-free fair division belongs to $\Omega(n!)$? Can we get a lower bound for a symmetric and envy-free fair division?\\

In the approximate setting we get  an $\varepsilon$-symmetric and $\varepsilon$-envy free fair division algorithm thanks to the $\varepsilon$-perfect division proposed in \cite{Branzei}. In this case the number of queries is in $\bigO(n^2/\varepsilon)$. However, this algorithm uses an exponential number in $n$ of arithmetic operations and inequality tests.\\ 
This problem appears also in our last algorithm which computes a symmetric and proportional fair division with  $\bigO(n^3)$ queries in the Robertson-Webb model. In this algorithm we solve a sub-problem (the computation of the set $S$) with an exponential number (in $n$) arithmetic operations and inequality tests.\\

Thus in these kinds of situations (Symmetric and Proportionnal or $\epsilon$-perfect fair division) an algorithm  with a polynomial number of queries cannot be considered as a fast algorithm if it uses an exponential number in $n$ of elementary operations.\\
A new model of computation has been suggested in \cite{ChezeBSSRW}. In this model the number of elementary operations must be taken into account. Then, in this model, the algorithm \texttt{SymProp} has not a polynomial complexity and cannot be considered as a fast algorithm. However, in a recent work based on a preprint of this article, Aigner-Horev and Segal-Halevi have shown how to modify \texttt{SymProp} in order to get an algorithm with a polynomial complexity even if we take into account the number of arithmetic operations, see \cite{Segal-Halevi-bipartite}. \\

We have constructed in this article an algorithm giving an aristotelian and proportional algorithm. This algorithm uses less arithmetic operations than  our symmetric and proportional algorithm but these two algorithms use the same number of queries. Is it necessary?\\

At last, the aristotelian notion comes from the Nichomachean Ethics by Aristotle and  one of the contributions of this article is to prove that we can compute an  aristotelian and proportional fair division efficiently (with a polynomial number of queries). This result is interesting since until now all aristotelian proportional fair division algorithms were envy-free algorithms and thus have a huge complexity in the Roberston-Webb model. However, another philosopher, Seneca, would have given a sever conclusion about this work:

\emph{``The mathematician teaches me how to lay out the dimensions of my estates; but I should rather be taught how to lay out what is enough for a man to own.[\ldots] What good is there for me in knowing how to parcel out a piece of land, if I know not how to share it with my brother? [\ldots] The mathematician teaches me how I may lose none of my boundaries; I, however, seek to learn how to lose them all with a light heart."}\\
 Letters Lucilius/Letter 88; Seneca.\\

\textbf{Acknowledgement}: The author thanks Erel Segal-Halevi for pointing to him an imprecision at the end of the algorithm \texttt{SymProp} in a previous version of this work.\\
At last, the author thanks \'Em\`elie, \'Elo\"ise and Timoth\'e for having implicitly suggested to study this problem.

 

\end{document}